%% file: main.tex
\documentclass[11pt]{article}
\usepackage{enumitem}
\usepackage{amsthm}
\usepackage{amssymb, amsmath, braket, xcolor}
\usepackage{geometry}
\usepackage[numbers,square,sort&compress]{natbib}
\usepackage{url}
\usepackage[hidelinks]{hyperref}
\usepackage{xspace}
\usepackage{iftex}
\usepackage{shellesc}

\newcommand{\Lean}{\textsc{Lean}~4\xspace}
\newcommand{\Mathlib}{\textsc{Mathlib}\xspace}
\ifPDFTeX
  \PackageError{Formalizing CHSH Rigidity in Lean 4}
    {Lean code blocks require XeLaTeX or LuaLaTeX}
\else
  \usepackage{fontspec}
  \defaultfontfeatures{Ligatures=TeX}
\fi

\usepackage[frozencache]{minted}
\definecolor{leanbg}{RGB}{255,255,255}
\setminted{
  fontsize=\footnotesize,
  breaklines,
  autogobble,
  tabsize=2,
  encoding=utf8
}
\newminted[leancode]{lean4}{
  bgcolor=leanbg,
  frame=lines,
  framesep=2mm,
  linenos,
  numbersep=6pt
}

\title{Formalizing CHSH Rigidity in Lean 4}
\author{
Tianrun Zhao\\
Stony Brook University
\and
Nengkun Yu\\
Stony Brook University
}

\newtheorem{theorem}{Theorem}[section]

\theoremstyle{definition}

\theoremstyle{definition}

\theoremstyle{remark}

\begin{document}
\maketitle

\begin{abstract}
Violation of the Clauser-Horne-Shimony-Holt (CHSH) inequality~\cite{clauser1969proposed} certifies genuine quantum correlations. In this work, we formalize in \Lean the rigidity theorem---any strategy achieving near-optimal CHSH value must be locally isometric to the canonical qubit strategy. In the course of formalization, we identified a gap in the argument of~\cite{McKague_2012}.
\end{abstract}

\input{section_02_introduction}

\input{section_08_Gap}
\input{section_07_case_study_i_rigidity_self_testing}

\input{section_04_proof}
\input{section_05_formalization}

\input{section_11_conclusion_and_future_work}

\section*{Acknowledgments}
We thank Valerio Scarani for his helpful suggestions.

\bibliographystyle{plainnat}
\bibliography{bibliography}
\end{document}

%% file: section_02_introduction.tex
\section{Introduction}

The Clauser-Horne-Shimony-Holt (CHSH) inequality is one of the central nonlocality statements in
quantum information theory~\cite{clauser1969proposed}. It gives a concrete, experimentally
accessible separation between classical and quantum correlations. In particular, CHSH violations underlie some of the best-known applications of
nonlocality, including quantum cryptography and self-testing~\cite{mayers2003self, PhysRevLett.67.661}.

In practice, we care about robustness.
For the CHSH game, the guiding principle is that if a bipartite strategy achieves a value close to
Tsirelson's bound \(2\sqrt2\), then, up to local isometries, the strategy must be close to the
canonical two-qubit realization consisting of an EPR pair together with the standard CHSH
observables, possibly tensored with an additional junk system. Making this statement fully precise
requires several nontrivial ingredients: spectral estimates for the ideal CHSH operator, exact and
approximate intertwining relations for extracted observables, careful regrouping of tensor factors,
and quantitative norm bounds. Many of these steps are standard in the literature, but they are
often presented tersely, which makes them a natural target for formal verification~\cite{de2015lean}.

Our work is carried out in \Lean, an interactive theorem prover in which mathematical objects,
definitions, and proofs are represented in a single formal language and checked by a small trusted
kernel~\cite{de2015lean}. This style of machine-checked mathematics is valuable not only because it
ensures the correctness of each theorem, but also because it forces every intermediate claim,
type conversion, and side condition to be made explicit. The surrounding \Mathlib ecosystem
provides a large reusable library of formalized mathematics~\cite{mathlib2020}, making it possible
to build substantial developments on top of shared infrastructure rather than formalizing each
technical component from scratch.

We formalize the above result modularly:
\[
    \underbrace{\text{CHSH near-optimality}}_{\text{bias assumption}}
    \;\longrightarrow\;
    \underbrace{\text{ideal expectation bound}}_{\text{extracted qubit analysis}}
    \;\longrightarrow\;
    \underbrace{\text{Bell-state overlap}}_{\text{state extraction}}
    \;\longrightarrow\;
    \underbrace{\text{operator extraction}}_{\text{local observable control}}.
\]

Their precise meaning of each step is clarified in the proof section.

By formalizing these stages as separate, composable components in \Lean, we obtain a
development in which the main rigidity theorem is assembled from reusable lemmas rather than from a
single monolithic proof. The present paper focuses on the robust CHSH rigidity theorem as the main
case study and uses it to illustrate both the proof architecture and the formalization techniques
needed to make each step precise.
This work can be viewed as a
formalization of a strengthened version of the 31st theorem (Tsirelson's bound) in ~\cite{david_palsberg_top100_quantum_theorems}.
Our \Lean formalization used in this work is available at \href{https://github.com/trzstony/RigidityTheorem}{\texttt{the project repository}}.

More broadly, we view formal verification as an important methodological tool for making results in
quantum information theory more reliable. Arguments in this area often involve long chains of
calculations, where a
tiny error can propagate and
ultimately undermine the final conclusion. Such issues can be hard to detect, because each individual step may look routine. A
proof assistant force every step to be stated explicitly. In this way, formal verification can reveal subtle mistakes that
might otherwise go unnoticed.

%% file: section_08_Gap.tex
\section{A Gap in the Original Literature}

The authors of~\cite{McKague_2012} define Bob-side operators as
\(X'_B := (B_0+B_1)/\lvert B_0+B_1\rvert\) and \(Z'_B := (B_0-B_1)/\lvert B_0-B_1\rvert\),
where \(\lvert\cdot \rvert\) is the modulus defined as \(\lvert M\rvert = \sqrt{M^2}\).
When \(B_0\pm B_1\) is not invertible (and therefore \((B_0\pm B_1)/\lvert B_0\pm B_1\rvert\) is not well defined), one must setup a convention.

The authors commented in the footnote: "If \(M\) has a subspace with eigenvalue \(0\), then the eigenvalue of \(M/\lvert M\rvert\) in that subspace is taken to be \(1\)." This is equivalent to saying "\(M/\lvert M\rvert\) acts as identity on \(\operatorname{ker}(M)\)."

Later, the authors said: "Moreover, \(\{X'_B,Z'_B\} = 0\) by construction, ...", which is incorrect. We show that in general \(\{X'_B,Z'_B\} \neq 0\) by constructing a counterexample below.

\paragraph{Concrete counterexample (in dimension 2).}
Let \(B_0=B_1=\sigma_z\).
Then \(B_0-B_1=0\), so under the ``+1 on the kernel'' convention, \[Z'_B:= (B_0-B_1)/\lvert B_0-B_1\rvert=I\] 

Since \(B_0 + B_1 = 2\sigma_z\), we have
\(
\lvert B_0+B_1\rvert=\sqrt{(2\sigma_z)^2}= 2I,
\)
so 
\[X'_B := (B_0+B_1)/\lvert B_0+B_1\rvert= \sigma_z.\]
Hence the anticommutator is
\(\{X'_B,Z'_B\}=2\sigma_z\ne 0\).

\paragraph{Comment}
Our \Lean implementation avoids relying on any kernel convention.
Instead, we extract Bob's qubit using a conjugation trick: we define a single-qubit rotation \(R\) that maps the ideal Pauli basis to the ideal CHSH basis (so \(RZR^{\!*}=H\) and \(RXR^{\!*}=H'\)), and we build Bob's unitary \(U_B\) by conjugating the same controlled-reflection circuit used on Alice's side.
This keeps the extracted operators unitary by construction.

The condition in the above counterexample is quite restrictive, and with the choice $B_0 = B_1 = \sigma_z$, the bias is at most 2. Therefore, we may assume $\epsilon$ is sufficiently small to fix this gap.

%% file: section_07_case_study_i_rigidity_self_testing.tex
\section{Rigidity and Self-Testing}

In this section, we formalize the robust CHSH rigidity theorem. Our proof follows Cleve's notes~\cite{cleve2019qic890} with some modifications.

\subsection{Bell states.}
On two qubits $\mathbb{C}^2$, we use the standard Bell basis
\[
\ket{\Phi^+} := \frac{\ket{00}+\ket{11}}{\sqrt{2}},\qquad
\ket{\Phi^-} := \frac{\ket{00}-\ket{11}}{\sqrt{2}},
\]
\[
\ket{\Psi^+} := \frac{\ket{01}+\ket{10}}{\sqrt{2}},\qquad
\ket{\Psi^-} := \frac{\ket{01}-\ket{10}}{\sqrt{2}}.
\]

These four states form an orthonormal basis of the two-qubit space.

\subsection{CHSH operators and strategies.}

\paragraph{Abstract CHSH strategy.}
A CHSH strategy consists of a bipartite state together with two binary observables for Alice and
two for Bob~\cite{cleve2004consequences}.  Concretely, let
\[
S=(\ket{\psi},A_0,A_1,B_0,B_1)
\]
be a CHSH strategy on finite-dimensional Hilbert spaces \(H_A \otimes H_B\), where
\(A_0,A_1\) and \(B_0,B_1\) are binary observables.

Its associated CHSH operator is
\[
\operatorname{CHSH}:=A_0\otimes B_0 + A_0\otimes B_1 + A_1\otimes B_0 - A_1\otimes B_1,
\]

\paragraph{Tsirelson's bound.}
Before turning to rigidity, recall the basic extremal fact~\cite{cirel1980quantum}: for any state \(\psi\) and binary
observables \(A_0,A_1,B_0,B_1\), the CHSH expectation is at most \(2\sqrt{2}\).  A very short
proof is to define
\[
P := \frac{A_0+A_1}{\sqrt{2}} \otimes I - I \otimes B_0,
\qquad
Q := \frac{A_1-A_0}{\sqrt{2}} \otimes I + I \otimes B_1.
\]
Then
\[
2\sqrt{2}\,I - \operatorname{CHSH}(A_0,A_1,B_0,B_1)
= \frac{1}{\sqrt{2}}\,(P^\dagger P + Q^\dagger Q)\ge 0,
\]
so taking expectation in \(\psi\) gives the bound. 

\Mathlib already contains this result in
\texttt{Mathlib/Algebra/Star/CHSH.lean} \cite{mathlib2020}.  Rigidity begins from the near-equality case of this
inequality.

\paragraph{Canonical two-qubit model.}
To identify the extremal strategy, we compare with the standard qubit observables.  Let \(Z\)
and \(X\) denote the Pauli matrices \(\sigma_z\) and \(\sigma_x\), and define
\[
H:=\frac{Z+X}{\sqrt2},\qquad H':=\frac{Z-X}{\sqrt2}.
\]

We then specialize to the canonical two-qubit choice
\[
A_0 := Z,\qquad A_1 := X,\qquad
B_0 := H,\qquad
B_1 := H'.
\]
The resulting two-qubit CHSH operator will be denoted by \(K\); explicitly,
\[
\begin{aligned}
K
:&=
Z\otimes H + Z\otimes H' + X\otimes H - X\otimes H'\\
&= \sqrt{2}\,(Z \otimes Z + X \otimes X).
\end{aligned}
\]

\paragraph{Spectral decomposition.}
The Bell basis diagonalizes this canonical operator.  In the basis
\(\{\ket{\Phi^+},\ket{\Phi^-},\ket{\Psi^+},\ket{\Psi^-}\}\), one has
\[
K
=
2\sqrt{2}\,\ket{\Phi^+}\!\bra{\Phi^+}
-2\sqrt{2}\,\ket{\Psi^-}\!\bra{\Psi^-},
\]
with \(\ket{\Phi^-}\) and \(\ket{\Psi^+}\) spanning the \(0\)-eigenspace. Thus
\(\ket{\Phi^+}\) is the unique eigenvector of eigenvalue \(2\sqrt{2}\), while
\(\ket{\Psi^-}\) is an eigenvector of eigenvalue \(-2\sqrt{2}\).

\subsection{Robust rigidity of CHSH}

We now consider the approximately extremal case of CHSH.  The point is no longer just that
\(2\sqrt{2}\) is the largest possible CHSH value, but that every strategy whose value is
close to \(2\sqrt{2}\) must already contain the standard EPR strategy, up to local isometries
and an additional irrelevant ``junk'' register.

We define the bias as
\[
\beta(S)
:=
\bra{\psi} \operatorname{CHSH}\ket{\psi}.
\]

\begin{theorem}[Robust CHSH rigidity~\cite{cleve2019qic890}]\label{thm:robust-chsh-rigidity}

Let $S = (|\psi\rangle, A_0, A_1, B_0, B_1)$ be an entangled strategy for the CHSH game.

Let \(\varepsilon\ge 0\) be sufficiently small, suppose the bias satisfies
$$
\beta(S) \ge 2\sqrt{2}-\varepsilon.
$$

Then there exist local isometries $V_A:H_A\to \mathbb{C}^2\otimes H_A$ and $V_B:H_B\to \mathbb{C}^2\otimes H_B$ and a state $|\Phi_{\text{junk}}\rangle\in H_A\otimes H_B$ such that:
\begin{enumerate}
\item \textbf{State extraction.}
$$
\|(V_A\otimes V_B)|\psi\rangle - |\Phi^+\rangle\otimes |\Phi_{\text{junk}}\rangle\|
\in O(\sqrt\varepsilon),
\qquad
|\Phi^+\rangle:=\frac{|00\rangle+|11\rangle}{\sqrt2}.
$$

\item \textbf{Operator extraction for Alice.}
$$
\|(V_A\otimes V_B)(A_0\otimes I)|\psi\rangle - (Z\otimes I)|\Phi^+\rangle\otimes |\Phi_{\text{junk}}\rangle\|
\in O(\sqrt\varepsilon),
$$
$$
\|(V_A\otimes V_B)(A_1\otimes I)|\psi\rangle - (X\otimes I)|\Phi^+\rangle\otimes |\Phi_{\text{junk}}\rangle\|
\in O(\sqrt\varepsilon).
$$

\item \textbf{Operator extraction for Bob.}
$$
\|(V_A\otimes V_B)(I\otimes B_0)|\psi\rangle - (I\otimes H)\,|\Phi^+\rangle\otimes |\Phi_{\text{junk}}\rangle\|
\in O(\sqrt\varepsilon),
$$
$$
\|(V_A\otimes V_B)(I\otimes B_1)|\psi\rangle - (I\otimes H')\,|\Phi^+\rangle\otimes |\Phi_{\text{junk}}\rangle\|
\in O(\sqrt\varepsilon).
$$
\end{enumerate}
\end{theorem}
The formalization proves explicit inequalities with concrete constants (omitted here for readability).

This is the robust self-testing statement for CHSH.  In the abstract CHSH framework, a strategy
is just a state together with four binary observables.

 The rigidity theorem specializes this
framework to the canonical target strategy
\[
A_0=Z,\qquad A_1=X,\qquad
B_0=H,\qquad B_1=H',
\]
and asserts that every near-optimal strategy must simulate this one approximately.

%% file: section_04_proof.tex
\section{Proof}

\begin{proof}[Proof Sketch]
We only sketch the proof formalized in our \Lean implementation. The main idea of the proof is still from Cleve's notes~\cite{cleve2019qic890}, but the formalization also makes one Bob-side modification explicit: instead of first working in the symmetrized CHSH basis on Bob's side and only absorbing the final rotation at the end, we build that rotation directly into Bob's extractor from the start.

\textbf{Setup: Constructing the local isometries \(V_A, V_B\).}

Following Cleve's construction on Alice's side, and with a modified Bob-side extractor in our formalization, we attach a qubit ancilla on each side and build extraction
unitaries from controlled versions of the observables together with the standard single-qubit gates
\(H,X,Z\). The resulting unitaries are
\[
U_A=C_{A_1}(H\otimes I)C_{A_0}(H\otimes I),
\qquad
U_B=(R\otimes I)\,C_{B_1}(H\otimes I)C_{B_0}(H\otimes I)\,(R^{\dagger}\otimes I),
\]
where \(C_A\) denotes the controlled version of \(A\), and
\[
R=\sin(\pi/8)\,X+\cos(\pi/8)\,Z,
\qquad
RZR^{\dagger}=H,
\qquad
RXR^{\dagger}=H',
\qquad
\ket{\mathrm{aux}}=R\ket0.
\]
The local isometries \(V_A : H_A \to \mathbb C^2\otimes H_A\) and
\(V_B : H_B \to \mathbb C^2\otimes H_B\) are then defined by
\[
V_A:=U_A(\ket0\otimes I_{H_A}),
\]
\[
V_B:=U_B(\ket{\mathrm{aux}}\otimes I_{H_B}).
\]
This Bob-side definition is one of our explicit modifications to the proof presentation. In Cleve's notes, the proof first produces the symmetrized canonical form on Bob's extracted qubit and only in the final step absorbs the local rotation \(R\) into Bob's isometry. Here we instead incorporate \(R\) into both \(U_B\) and the ancilla state \(\ket{\mathrm{aux}}=R\ket0\) from the outset. This lets the proof target the standard CHSH observables \((H,H')\) throughout, rather than switching conventions at the end.

It remains to construct the junk state \(\ket{\Phi_{\mathrm{junk}}}\) and show that these
isometries approximately extract the ideal Bell-pair structure.

From
\[
\beta:=\langle\psi|A_0\otimes B_0+A_0\otimes B_1+A_1\otimes B_0-A_1\otimes B_1|\psi\rangle
\ge 2\sqrt2-\varepsilon
\]
Let \(c:=128\sqrt2\). Then we can obtain
\[
\bra{\psi}(A_0A_1+A_1A_0)^2\otimes I \ket{\psi} \le c\,\varepsilon,
\qquad
\bra{\psi}I \otimes (B_0B_1+B_1B_0)^2\ket{\psi} \le c\,\varepsilon.
\]

\textbf{(1) State extraction.}
Set
\[
\ket{\Psi}:=\operatorname{regSwap}\bigl((V_A\otimes V_B)\ket{\psi}\bigr)
\in (\mathbb C^2\otimes\mathbb C^2)\otimes(H_A\otimes H_B),
\]
where \(\operatorname{regSwap}\) is the canonical linear isomorphism that regroups the tensor
factors.  This is needed because \((V_A\otimes V_B)\ket{\psi}\) naturally lies in
\((\mathbb C^2\otimes H_A)\otimes(\mathbb C^2\otimes H_B)\), whereas
\(\ket{\Phi^+}\otimes\ket{\Phi_{\mathrm{junk}}}\) lies in
\((\mathbb C^2\otimes\mathbb C^2)\otimes(H_A\otimes H_B)\).
This regrouping map is left implicit in Cleve's notes, but in our formalization it must be made
explicit: one of the small but genuine contributions of the development is to isolate such tensor
reassociation steps as concrete maps, since a proof assistant requires every change of ambient space
to be represented rigorously.

Let
\[
K=(Z\otimes H)+(Z\otimes H')+(X\otimes H)-(X\otimes H').
\]
This is exactly the canonical two-qubit CHSH operator introduced above, now acting on the
extracted qubit registers.
\[
\operatorname{CHSH}_{\mathrm{phys}}
:=\operatorname{regSwap}\bigl((V_A\otimes V_B)\operatorname{CHSH}\ket{\psi}\bigr),
\qquad
\operatorname{CHSH}_{\mathrm{ideal}}:=(K\otimes I)\ket{\Psi}.
\]
Term-by-term, one has
\[
\|\operatorname{CHSH}_{\mathrm{ideal}}-\operatorname{CHSH}_{\mathrm{phys}}\|
\le 4\sqrt{c\varepsilon},
\]

Let
\(
\delta:=\varepsilon+4\sqrt{c\varepsilon}\).
Since \(\ket{\Psi}\) is a unit vector, the preceding norm bound implies
\[
\Re\langle \Psi | \operatorname{CHSH}_{\mathrm{ideal}} \rangle
\ge
\Re\langle \Psi | \operatorname{CHSH}_{\mathrm{phys}} \rangle - 4\sqrt{c\varepsilon}.
\]
Using \(\operatorname{CHSH}_{\mathrm{ideal}}=(K\otimes I)\ket{\Psi}\),
\(\Re\langle \Psi | \operatorname{CHSH}_{\mathrm{phys}} \rangle=\beta\), and
\(\beta\ge 2\sqrt2-\varepsilon\), we obtain
\[
\Re\langle \Psi | (K\otimes I) | \Psi\rangle \ge 2\sqrt2-\delta.
\]
Using the spectral decomposition of \(K\) in the Bell basis, the preceding inequality implies that
the projection of \(\ket{\Psi}\) onto the top eigenspace of \(K\), namely the
\(\ket{\Phi^+}\)-eigenspace, has squared norm at least \(1-\delta/(2\sqrt2)\). Let
\(\ket{\Phi_{\mathrm{junk}}}\) be the normalized component of \(\ket{\Psi}\) in that eigenspace.
Then the extracted state is close to the ideal Bell pair tensored with junk:
\[
\|\ket{\Psi}-\ket{\Phi^+}\otimes\ket{\Phi_{\mathrm{junk}}}\|
\le \sqrt{\delta/\sqrt2}.
\]

\textbf{(2) Operator extraction for Alice.}
For \(A_0\), the intertwining relation is exact, while for \(A_1\) we use the approximate
extraction bound coming from the anticommutator estimate.  Applying these identities to
\(\ket{\psi}\), and then transporting them through \(\operatorname{regSwap}\), gives
\[
(Z\otimes I)V_A = V_AA_0,
\qquad
\Bigl\| \bigl((((X\otimes I)V_A - V_AA_1)\otimes I)\ket{\psi}\bigr) \Bigr\|
\le \sqrt{c\varepsilon}.
\]
Combining these relations with the state-extraction estimate and using that the ideal observables
act isometrically on the extracted qubit registers, we obtain the two required Alice-side bounds:
\[
\Bigl\|\operatorname{regSwap}\bigl((V_A\otimes V_B)(A_0\otimes I)\ket{\psi}\bigr)
-((Z\otimes I)\otimes I)\bigl(\ket{\Phi^+}\otimes\ket{\Phi_{\mathrm{junk}}}\bigr)\Bigr\|
\le \sqrt{\delta/\sqrt2},
\]
\[
\Bigl\|\operatorname{regSwap}\bigl((V_A\otimes V_B)(A_1\otimes I)\ket{\psi}\bigr)
-((X\otimes I)\otimes I)\bigl(\ket{\Phi^+}\otimes\ket{\Phi_{\mathrm{junk}}}\bigr)\Bigr\|
\le \sqrt{c\varepsilon}+\sqrt{\delta/\sqrt2}.
\]

\textbf{(3) Operator extraction for Bob.}
The argument for Bob is identical in structure, but here the Bob-side modification above becomes mathematically visible. Because we built the rotation \(R\) and the rotated ancilla \(\ket{\mathrm{aux}}\) into the definition of \(V_B\), the extracted \(B_0\) relation already matches the standard observable \(H\) exactly at this stage, while \(H'\) is extracted approximately.
\[
\!(H\otimes I)V_B = V_BB_0,
\qquad
\Bigl\| \bigl((I\otimes ((H'\otimes I)V_B - V_BB_1))\ket{\psi}\bigr) \Bigr\|
\le \sqrt{c\varepsilon}.
\]
This is precisely where our presentation differs from Cleve's notes: there, the analogous exact intertwining is first obtained in the symmetrized basis, and the \(R\)-rotation is absorbed only in the final cleanup step. In our proof sketch and Lean implementation, the stronger-looking statement
\((H\otimes I)V_B = V_BB_0\)
already holds before that final assembly step because the extractor was designed in the rotated convention from the beginning.
Combining these with the same state-extraction estimate yields the Bob-side bounds:
\[
\Bigl\|\operatorname{regSwap}\bigl((V_A\otimes V_B)(I\otimes B_0)\ket{\psi}\bigr)
-((I\otimes H)\otimes I)\bigl(\ket{\Phi^+}\otimes\ket{\Phi_{\mathrm{junk}}}\bigr)\Bigr\|
\le \sqrt{\delta/\sqrt2},
\]
\[
\Bigl\|\operatorname{regSwap}\bigl((V_A\otimes V_B)(I\otimes B_1)\ket{\psi}\bigr)
-((I\otimes H')\otimes I)\bigl(\ket{\Phi^+}\otimes\ket{\Phi_{\mathrm{junk}}}\bigr)\Bigr\|
\le \sqrt{c\varepsilon}+\sqrt{\delta/\sqrt2}.
\]
Since \(\operatorname{regSwap}\) is an isometry, these are equivalent to the four bounds stated in
Theorem~\ref{thm:robust-chsh-rigidity}.
\end{proof}

%% file: section_05_formalization.tex
\section{Formalization of robust CHSH rigidity}

\paragraph{Representative \Lean declarations.}
The following declarations from \texttt{RigidityTheorem/Approximate\_Rigidity/} capture the main
formal ingredients behind the three parts of the proof architecture.

\textbf{Strategy and observables.}
The basic CHSH data are packaged by the structure \texttt{CHSHStrategy}, whose fields are the
shared state \(\psi\), the four observables \(A_0,A_1,B_0,B_1\), and the normalization condition
\(\|\psi\|=1\). The class \texttt{IsBinaryObservable} formalizes the assumption that each
observable is a self-adjoint involution, and the definition \texttt{chshBias} is the real part of
the expectation of the abstract CHSH operator \texttt{CHSH\_op} on this strategy.

\begin{leancode}
/-! ## CHSH strategy scaffolding -/

structure CHSHStrategy (H_A H_B : Type*)
    [NormedAddCommGroup H_A] [NormedAddCommGroup H_B]
    [InnerProductSpace ℂ H_A] [InnerProductSpace ℂ H_B] where
  psi : H_A ⊗[ℂ] H_B
  psi_norm : ‖psi‖ = 1
  A0 : H_A →ₗ[ℂ] H_A
  A1 : H_A →ₗ[ℂ] H_A
  B0 : H_B →ₗ[ℂ] H_B
  B1 : H_B →ₗ[ℂ] H_B
  A0_bin : IsBinaryObservable A0
  A1_bin : IsBinaryObservable A1
  B0_bin : IsBinaryObservable B0
  B1_bin : IsBinaryObservable B1
\end{leancode}

\begin{leancode}
 /-- General CHSH operator. -/
def CHSH_op : H_A ⊗[ℂ] H_B →ₗ[ℂ] H_A ⊗[ℂ] H_B :=
  (A0 ⊗ₗ B0) + (A0 ⊗ₗ B1) + (A1 ⊗ₗ B0) - (A1 ⊗ₗ B1)
\end{leancode}
  
\begin{leancode} 
/-- Binary observable = self-adjoint involution. -/
class IsBinaryObservable {H : Type*} [NormedAddCommGroup H] [InnerProductSpace ℂ H]
    (A : H →ₗ[ℂ] H) : Prop where
  symm : A.IsSymmetric
  sq_one : A ∘ₗ A = LinearMap.id
\end{leancode}

\textbf{Controlled gates and ancilla embeddings.}
In the isometry construction, the definition \texttt{control} is exactly the controlled gate
\((|0\rangle\langle 0| \otimes I)+(|1\rangle\langle 1| \otimes A)\) on
\(\mathbb C^2 \otimes H\). The map \texttt{embed aux} is the ancilla insertion
\(\ket{\mathrm{aux}}\otimes I\), and \texttt{auxState} formalizes Bob's rotated ancilla
\(\ket{\mathrm{aux}}=R\ket0\).

\begin{leancode}
noncomputable def auxState : Qubit := Rotation ket0

-- Embedding `S_aux = |aux⟩ ⊗ I` (and in particular `S = |0⟩ ⊗ I` for Alice).
noncomputable def embed (aux : Qubit) : H →ₗ[ℂ] (Qubit ⊗[ℂ] H) :=
  (TensorProduct.mk ℂ Qubit H) aux
\end{leancode}
\newpage
\begin{leancode}
noncomputable def control (A : H →ₗ[ℂ] H) :
    (Qubit ⊗[ℂ] H) →ₗ[ℂ] (Qubit ⊗[ℂ] H) :=
  -- Projectors `|0⟩⟨0|` and `|1⟩⟨1|` on the control qubit.
  -- Controlled gate: `(|0⟩⟨0| ⊗ I) + (|1⟩⟨1| ⊗ A)`.
  (proj0 ⊗ₗ (LinearMap.id : H →ₗ[ℂ] H)) + (proj1 ⊗ₗ A)
\end{leancode}

\textbf{The maps \texttt{unitaryUA}, \texttt{unitaryUB}, \texttt{VA}, and \texttt{VB}.}
The definition \texttt{unitaryUA} is the \Lean version of the canonical circuit
\[
U_A=C_{A_1}(H\otimes I)C_{A_0}(H\otimes I),
\]
built from two controlled gates and two Hadamards. The isometry \texttt{VA} is then obtained by
composing \texttt{unitaryUA} with the embedding \(\ket0\otimes I\), so it is the formal map
\(V_A=U_A(\ket0\otimes I)\). For Bob, the definition \texttt{unitaryUB} implements the rotated
construction
\[
U_B=(R\otimes I)\,U_A(B_0,B_1)\,(R^\dagger\otimes I),
\]
and \texttt{VB} is the corresponding embedded map \(V_B=U_B(\ket{\mathrm{aux}}\otimes I)\).

\begin{leancode}
/- Canonical form unitaries built from Alice/Bob observables.-/
noncomputable def unitaryUA (A0 A1 : H →ₗ[ℂ] H) :
    (Qubit ⊗[ℂ] H) →ₗ[ℂ] (Qubit ⊗[ℂ] H) :=
  (control A1) ∘ₗ (Hadamard ⊗ₗ LinearMap.id) ∘ₗ (control A0) ∘ₗ (Hadamard ⊗ₗ LinearMap.id)

noncomputable def VA (A0 A1 : H →ₗ[ℂ] H) :
    H →ₗ[ℂ] (Qubit ⊗[ℂ] H) :=
  (unitaryUA A0 A1).comp (embed ket0)
\end{leancode}

\begin{leancode}
/-!
Bob's unitary in canonical form (rewritten to match the convention
`(B0,B1) = (H,H')` in the notes):

`U_B := (R ⊗ I) · U_A(B0, B1) · (R† ⊗ I)` where `R = sin(π/8) X + cos(π/8) Z`.

This choice ensures that, with the embedding `S_aux = |aux⟩ ⊗ I` where `|aux⟩ = R|0⟩`,
the extracted action on `B0` is exact.
-/
noncomputable def unitaryUB (B0 B1 : H →ₗ[ℂ] H) :
    (Qubit ⊗[ℂ] H) →ₗ[ℂ] (Qubit ⊗[ℂ] H) :=
  (Rotation ⊗ₗ (LinearMap.id : H →ₗ[ℂ] H)) ∘ₗ
    (unitaryUA B0 B1) ∘ₗ
    (Rotation.adjoint ⊗ₗ (LinearMap.id : H →ₗ[ℂ] H))

noncomputable def VB (B0 B1 : H →ₗ[ℂ] H) :
    H →ₗ[ℂ] (Qubit ⊗[ℂ] H) :=
  (unitaryUB B0 B1).comp (embed auxState)
\end{leancode}

\textbf{Tensor regrouping.}
The definition \texttt{regSwap} implements the canonical linear isomorphism between
\(((\mathbb C^2 \otimes H_A) \otimes (\mathbb C^2 \otimes H_B))\) and
\(((\mathbb C^2 \otimes \mathbb C^2) \otimes (H_A \otimes H_B))\).  It is the formal device that
allows the extracted physical state to be compared directly with
\(\ket{\Phi^+}\otimes\ket{\Phi_{\mathrm{junk}}}\).  Concretely, the code below builds this map by
two reassociations together with a swap of the middle registers.

\begin{leancode}
noncomputable def regSwap :
    ((Qubit ⊗[ℂ] H_A) ⊗[ℂ] (Qubit ⊗[ℂ] H_B)) →ₗ[ℂ]
      ((Qubit ⊗[ℂ] Qubit) ⊗[ℂ] (H_A ⊗[ℂ] H_B)) :=
by
  -- First reassociate so we can access the middle factors.
  let assoc₁ : ((Qubit ⊗[ℂ] H_A) ⊗[ℂ] (Qubit ⊗[ℂ] H_B)) →ₗ[ℂ]
      (Qubit ⊗[ℂ] (H_A ⊗[ℂ] (Qubit ⊗[ℂ] H_B))) :=
    (TensorProduct.assoc ℂ Qubit H_A (Qubit ⊗[ℂ] H_B)).toLinearMap
  -- Rearrange `H_A ⊗ (Qubit ⊗ H_B)` into `Qubit ⊗ (H_A ⊗ H_B)`.
  let inner : (H_A ⊗[ℂ] (Qubit ⊗[ℂ] H_B)) →ₗ[ℂ]
      (Qubit ⊗[ℂ] (H_A ⊗[ℂ] H_B)) :=
    (TensorProduct.assoc ℂ Qubit H_A H_B).toLinearMap.comp
      ((TensorProduct.map (TensorProduct.comm ℂ H_A Qubit).toLinearMap
          (LinearMap.id : H_B →ₗ[ℂ] H_B)).comp
        (TensorProduct.assoc ℂ H_A Qubit H_B).symm.toLinearMap)
  -- Tensor the above rearrangement with `id` on the leading `Qubit`.
  let mapMid : (Qubit ⊗[ℂ] (H_A ⊗[ℂ] (Qubit ⊗[ℂ] H_B))) →ₗ[ℂ]
      (Qubit ⊗[ℂ] (Qubit ⊗[ℂ] (H_A ⊗[ℂ] H_B))) :=
    TensorProduct.map (LinearMap.id : Qubit →ₗ[ℂ] Qubit) inner
  -- Finally reassociate to get `(Qubit ⊗ Qubit) ⊗ (H_A ⊗ H_B)`.
  let assoc₂ : (Qubit ⊗[ℂ] (Qubit ⊗[ℂ] (H_A ⊗[ℂ] H_B))) →ₗ[ℂ]
      ((Qubit ⊗[ℂ] Qubit) ⊗[ℂ] (H_A ⊗[ℂ] H_B)) :=
    (TensorProduct.assoc ℂ Qubit Qubit (H_A ⊗[ℂ] H_B)).symm.toLinearMap
  exact assoc₂.comp (mapMid.comp assoc₁)
\end{leancode}

\textbf{State extraction bridge.}
The theorem \texttt{chsh\_to\_K\_expectation} is the formal bridge from the physical CHSH bias to
the extracted state: it introduces the local isometries \(V_A,V_B\), defines
\(\Psi=\operatorname{regSwap}((V_A\otimes V_B)\psi)\), and proves that
\(\Re\langle \Psi,(K\otimes I)\Psi\rangle\) is at least \(2\sqrt2-\delta(\varepsilon)\).
This is the exact expectation estimate used before the Bell-basis analysis in the proof sketch.

\begin{leancode}
theorem chsh_to_K_expectation
    (hBias : chshBias S ≥ 2 * Real.sqrt 2 - ε) :
    let V_A : H_A →ₗ[ℂ] (Qubit ⊗[ℂ] H_A) := VA (H := H_A) S.A0 S.A1
    let V_B : H_B →ₗ[ℂ] (Qubit ⊗[ℂ] H_B) := VB (H := H_B) S.B0 S.B1
    let Ψ : (Qubit ⊗[ℂ] Qubit) ⊗[ℂ] (H_A ⊗[ℂ] H_B) := regSwap ((V_A ⊗ₗ V_B) S.psi)
    Complex.re
        (⟪Ψ, ((K ⊗ₗ (LinearMap.id : (H_A ⊗[ℂ] H_B) →ₗ[ℂ] (H_A ⊗[ℂ] H_B))) Ψ)⟫_ℂ)
      ≥ 2 * Real.sqrt 2 - delta ε := by ...
\end{leancode}

\textbf{Operator extraction for Alice.}
The theorem \texttt{Alice\_operator\_extraction} is the packaged Alice-side conclusion: it asserts
the existence of \(V_A,V_B\) and a junk state such that the extracted action of \(A_0\) is close
to \(Z\ket{\Phi^+}\) with error \(\sqrt{\delta(\varepsilon)/\sqrt2}\), while the extracted action
of \(A_1\) is close to \(X\ket{\Phi^+}\) with the larger error
\(\sqrt{c\varepsilon}+\sqrt{\delta(\varepsilon)/\sqrt2}\).

\begin{leancode}
theorem Alice_operator_extraction
    (hBias : chshBias S ≥ 2 * Real.sqrt 2 - ε) :
    ∃ V_A : H_A →ₗ[ℂ] (Qubit ⊗[ℂ] H_A),
      ∃ V_B : H_B →ₗ[ℂ] (Qubit ⊗[ℂ] H_B),
      ∃ junk : H_A ⊗[ℂ] H_B,
        Isometry V_A ∧
          Isometry V_B ∧
            ‖(regSwap ((V_A ⊗ₗ V_B) ((applyAlice S.A0) S.psi))
              - ((pauliZ ⊗ₗ LinearMap.id) bellState) ⊗ₜ[ℂ] junk)‖
              ≤ Real.sqrt (delta ε / Real.sqrt 2) ∧
            ‖(regSwap ((V_A ⊗ₗ V_B) ((applyAlice S.A1) S.psi))
              - ((pauliX ⊗ₗ LinearMap.id) bellState) ⊗ₜ[ℂ] junk)‖
              ≤ Real.sqrt (cConst * ε) + Real.sqrt (delta ε / Real.sqrt 2) := by ...	\end{leancode}

\textbf{Operator extraction for Bob.}
Similarly, \texttt{Bob\_operator\_extraction} is the packaged Bob-side conclusion: it states that
the extracted action of \(B_0\) is close to \(H\ket{\Phi^+}\) with error
\(\sqrt{\delta(\varepsilon)/\sqrt2}\), and the extracted action of \(B_1\) is close to
\(H'\ket{\Phi^+}\)---written in \Lean as \texttt{pauliZHZ}---with error
\(\sqrt{c\varepsilon}+\sqrt{\delta(\varepsilon)/\sqrt2}\).

\begin{leancode}
theorem Bob_operator_extraction
    (hBias : chshBias S ≥ 2 * Real.sqrt 2 - ε) :
    ∃ V_A : H_A →ₗ[ℂ] (Qubit ⊗[ℂ] H_A),
      ∃ V_B : H_B →ₗ[ℂ] (Qubit ⊗[ℂ] H_B),
      ∃ junk : H_A ⊗[ℂ] H_B,
        Isometry V_A ∧
          Isometry V_B ∧
            ‖(regSwap ((V_A ⊗ₗ V_B) ((applyBob S.B0) S.psi))
              - ((LinearMap.id ⊗ₗ Hadamard) bellState) ⊗ₜ[ℂ] junk)‖
              ≤ Real.sqrt (delta ε / Real.sqrt 2) ∧
            ‖(regSwap ((V_A ⊗ₗ V_B) ((applyBob S.B1) S.psi))
              - ((LinearMap.id ⊗ₗ pauliZHZ) bellState) ⊗ₜ[ℂ] junk)‖
              ≤ Real.sqrt (cConst * ε) + Real.sqrt (delta ε / Real.sqrt 2) := by ...	\end{leancode}

%% file: section_11_conclusion_and_future_work.tex
\section{Discussion}

\subsection{AI assistance in the development process}
We used LLM-based coding assistants to accelerate formalization, in the same general direction as
recent work on agentic autoformalization in quantum computation~\cite{ren2026merlean}. Sometimes,
the model would quietly add hypotheses to a lemma, and the proof became trivial. We choose to
state each lemma in \Lean syntax before attempting a proof, and then closing \texttt{sorry}s from
the leaves upward. LLM assistance proved most effective once the relevant algebraic structure was
already made explicit in the theorem statement.
ChatGPT (OpenAI) was also used as a writing assistant in the preparation of this manuscript.

\subsection{Library ecosystem and integration.}
Our development proceeded in parallel with the quantum-information library of~\cite{meiburg2025formalizationgeneralizedquantumsteins}. Several of its design choices were not well suited to the proof patterns we required, so rather than building on top of it we developed our own library needed for our case studies.

\subsection{Formalization as error discovery.}
Our experience is also similar in spirit to recent work of Tooby-Smith, who formalized the stability of the two Higgs doublet model potential in \Lean and identified an error in the literature~\cite{toobysmith2026formalizingstabilityhiggsdoublet}. Both projects illustrate that formalization can do more than certify accepted arguments: it can also expose subtle gaps in published proofs.

%% file: bibliography.bib
@article{McKague_2012,
doi = {10.1088/1751-8113/45/45/455304},
url = {https://doi.org/10.1088/1751-8113/45/45/455304},
year = {2012},
month = {oct},
publisher = {IOP Publishing},
volume = {45},
number = {45},
pages = {455304},
author = {McKague, M and Yang, T H and Scarani, V},
title = {Robust self-testing of the singlet},
journal = {Journal of Physics A: Mathematical and Theoretical}
}

@misc{cleve2019qic890,
  author = {Richard Cleve},
  title = {QIC 890 Entanglement and Nonlocal Effects},
  year = {2019},
  month = apr,
  note = {Lecture notes for QIC 890, University of Waterloo},
  url = {https://cleve.iqc.uwaterloo.ca/resources/Qic890LectureNotes2019Apr22(V22).pdf}
}

@misc{david_palsberg_top100_quantum_theorems,
  author       = {Marco David and Jens Palsberg},
  title        = {Top 100 Quantum Theorems},
  year         = {2026},
  url          = {https://marcodavid.net/top100/},
  note         = {Accessed 2026-03-26}
}

@misc{toobysmith2026formalizingstabilityhiggsdoublet,
      title={Formalizing the stability of the two Higgs doublet model potential into Lean: identifying an error in the literature}, 
      author={Joseph Tooby-Smith},
      year={2026},
      eprint={2603.08139},
      archivePrefix={arXiv},
      primaryClass={hep-ph},
      url={https://arxiv.org/abs/2603.08139}, 
}

@misc{meiburg2025formalizationgeneralizedquantumsteins,
      title={A Formalization of the Generalized Quantum Stein's Lemma in Lean}, 
      author={Alex Meiburg and Leonardo A. Lessa and Rodolfo R. Soldati},
      year={2025},
      eprint={2510.08672},
      archivePrefix={arXiv},
      primaryClass={quant-ph},
      url={https://arxiv.org/abs/2510.08672}, 
}

@inproceedings{mathlib2020,
  author    = {{The mathlib Community}},
  title     = {The {L}ean {M}athematical {L}ibrary},
  booktitle = {Proceedings of the 9th {ACM} {SIGPLAN} International Conference
               on Certified Programs and Proofs},
  series    = {CPP 2020},
  publisher = {ACM},
  address   = {New Orleans, LA, USA},
  year      = {2020},
  month     = jan,
  doi       = {10.1145/3372885.3373824},
  url       = {https://doi.org/10.1145/3372885.3373824}
}

@article{cirel1980quantum,
  title={Quantum generalizations of Bell's inequality},
  author={Cirel'son, Boris S},
  journal={Letters in Mathematical Physics},
  volume={4},
  number={2},
  pages={93--100},
  year={1980},
  publisher={Springer}
}

@article{PhysRevLett.67.661,
  title = {Quantum cryptography based on Bell's theorem},
  author = {Ekert, Artur K.},
  journal = {Phys. Rev. Lett.},
  volume = {67},
  issue = {6},
  pages = {661--663},
  numpages = {0},
  year = {1991},
  month = {Aug},
  publisher = {American Physical Society},
  doi = {10.1103/PhysRevLett.67.661},
  url = {https://link.aps.org/doi/10.1103/PhysRevLett.67.661}
}

@article{mayers2003self,
  title={Self testing quantum apparatus},
  author={Mayers, Dominic and Yao, Andrew},
  journal={arXiv preprint quant-ph/0307205},
  year={2003}
}

@article{clauser1969proposed,
  title={Proposed experiment to test local hidden-variable theories},
  author={Clauser, John F and Horne, Michael A and Shimony, Abner and Holt, Richard A},
  journal={Physical review letters},
  volume={23},
  number={15},
  pages={880},
  year={1969},
  publisher={APS}
}

@inproceedings{cleve2004consequences,
  title={Consequences and limits of nonlocal strategies},
  author={Cleve, Richard and Hoyer, Peter and Toner, Benjamin and Watrous, John},
  booktitle={Proceedings. 19th IEEE Annual Conference on Computational Complexity, 2004.},
  pages={236--249},
  year={2004},
  organization={IEEE}
}

@inproceedings{de2015lean,
  title={The Lean theorem prover (system description)},
  author={De Moura, Leonardo and Kong, Soonho and Avigad, Jeremy and Van Doorn, Floris and von Raumer, Jakob},
  booktitle={International Conference on Automated Deduction},
  pages={378--388},
  year={2015},
  organization={Springer}
}

@article{ren2026merlean,
  title         = {MerLean: An Agentic Framework for Autoformalization
                   in Quantum Computation},
  author        = {Ren, Yuanjie and Li, Jinzheng and Qi, Yidi},
  journal       = {arXiv preprint arXiv:2602.16554},
  year          = {2026},
  eprint        = {2602.16554},
  archivePrefix = {arXiv},
  primaryClass  = {cs.LO},
  doi           = {10.48550/arXiv.2602.16554}
}
